\documentclass[runningheads,a4paper]{llncs}
\usepackage{amssymb}
\usepackage{graphicx,subfigure}
\usepackage{color}
\usepackage[section,boxed]{algorithm}
\usepackage[noend]{algorithmic}
\usepackage{ifthen}

\newtheorem{observation}{Observation}

\usepackage{url}

\urldef{\mailt}\path|chhaya.dhingra@gmail.com|
\urldef{\mailk}\path|kutten@ie.technion.ac.il|
\newcommand{\keywords}[1]{\par\addvspace\baselineskip
\noindent\keywordname\enspace\ignorespaces#1}

\begin{document}

\mainmatter  

\title{Fast and Compact Distributed Verification and
              Self-Stabilization of a DFS Tree}

\titlerunning{Distributed Verification and Self-Stabilization of DFS}

%
%
\author{Shay Kutten
\and Chhaya Trehan}

\authorrunning{Kutten and Trehan}

\institute{Faculty of Industrial Engineering and Management, Technion, Haifa, Israel\\
\mailk, \mailt\\
}

%
%

\toctitle{Fast and Compact Distributed Verification and Self-Stabilization of a DFS Tree}
\tocauthor{Shay Kutten and Chhaya Trehan}
\maketitle

\setcounter{secnumdepth}{3}

\begin{abstract}
We present algorithms for distributed verification and silent-stabilization of a DFS(Depth First Search) spanning tree of a connected network. Computing and maintaining such a DFS tree is an important task, e.g., for constructing efficient routing schemes. Our algorithm improves upon previous work in various ways. Comparable previous work has space and time complexities of  $O(n\log \Delta)$ bits per node and $O(nD)$ respectively, where $\Delta$ is the highest degree of a node, $n$ is the number of nodes and $D$ is the diameter of the network. In contrast, our algorithm has a space complexity of $O(\log n)$ bits per node, which is optimal for silent-stabilizing spanning trees and runs in $O(n)$ time. In addition, our solution is modular since it utilizes the distributed verification algorithm as an independent subtask of the overall solution. It is possible to use the verification algorithm as a stand alone task or as a subtask in another algorithm. To demonstrate the simplicity of constructing efficient DFS 
algorithms using the modular approach, we also present a (non-silent) self-stabilizing DFS token circulation algorithm for general networks based on our silent-stabilizing DFS tree. The complexities of this token circulation algorithm are comparable to the known ones.
\end{abstract}

\section{Introduction}
A clear separation is common between the notions of computing and verification in sequential systems. A similar separation in the context of distributed systems has been emerging. Distributed verification of global properties like minimum spanning trees have been devised~\cite{MSTVerification}. 

An area of distributed systems that can greatly benefit from this separation is that of self-stabilization. Self-stabilization is the ability of a system to recover from transient faults. A self-stabilizing distributed system can be started in any arbitrary configuration and must eventually converge to a desired legal behavior.  Self-stabilizing algorithms can run a distributed verification algorithm repeatedly to detect the occurrence of faults in the system and take the necessary action for convergence to a legal behavior. This is the approach we take here in devising a silent-stabilizing DFS algorithm. The concept of   first detecting a fault and then taking the corrective measures for self-stabilization was first introduced by~\cite{KatzPerry},~\cite{Afek_thelocal} and~\cite{AwerbuchTransformer}. The approach taken by Katz and Perry in~\cite{KatzPerry} is that of global detection of faults by a leader node that periodically takes the \emph{snapshots} of the global state of the network and resets the 
system if a fault is detected. Afek, Kutten and Yung~\cite{Afek_thelocal}, and Awebuch et al.~\cite{AwerbuchTransformer} on the other hand, suggested that the faults in the global state of a system could sometimes be detected by local means - i.e., by having each node check the states of all its neighbors. G{\"{o}}{\"{o}}s and Suomela  further formalized the idea of local detection of faults in~\cite{SuomelaGoos}. Korman, Kutten and Peleg~\cite{ProofLabeling} introduced the concept of \emph{proof labeling schemes}. A \emph{proof labeling scheme} works by assigning a \emph{label} to every node in the input network. The collection of labels assigned to the nodes acts as a locally checkable \emph{distributed proof} that the global state of the network satisfies a specific global predicate. A \emph{proof labeling scheme} consists of a pair of algorithms $(\mathcal{M}$, $\mathcal{V})$, where $\mathcal{M}$ is a \emph{marker} algorithm that generates a label for every node and $\mathcal{V}$ is a \emph{verifier} 
algorithm that checks the \emph{labels} of neighboring nodes. In this paper, we present a \emph{proof labeling scheme} for detecting faults in the distributed representation of a DFS spanning tree. For self-stabilization, the DFS tree is computed afresh and new labels are assigned to the nodes by the marker on detection of faults.

\subsection{Additional Related Work}
Dijkstra introduced the concept of Self-stabilization~\cite{Dijkstra:1974} in distributed systems. Self-stabilization deals with the \emph{faults} that entail an arbitrary corruption of the state of a system. These faults are rather severe in nature but do not occur very frequently in     reality~\cite{Varghese00thefault}. \\
Table~\ref{table: related work} summarizes the known complexity results for self-stabilizing DFS algorithms. Collin and Dolev presented a silent-stabilizing DFS tree algorithm in~\cite{SelfStabDFS}. Their algorithm works by having each node store its path to the root node in the DFS tree. Since the path of a node to the root in a DFS tree can be as long as $n$, the number of nodes in the network, the space complexity of their algorithm is $O(n\log \Delta)$ per node, where $\Delta$ is the highest degree of a node in the network. The time complexity of their algorithm under the \emph{contention} time model is $(nD\Delta)$. We drop the multiplicative factor of $\Delta$ from their time complexity here for the sake of comparison with all the other algorithms that do not count their time under the \emph{contention} model. Cournier et al. presented a snap-stabilizing DFS \emph{wave} protocol in~\cite{SnapStabCournier} which snap stabilizes with a space complexity of $O(n\log n)$.

\begin{table}[h!]

\begin{center}

    \begin{tabular}{|l|l|l|p{5cm}|}\label{table: related work} 
    
    Algorithm & Space & Stabilization Time & Remarks \\ \hline
    ~\cite{SelfStabDFS} & $O(n \log \Delta)$ & $O(nD)$ & Silent \\ \hline
    ~\cite{SnapStabCournier} & $O(n\log n)$ & 0 & Snap Stabilizing  \\ & & & \emph{first} DFS wave, needs Unique IDs \\ \hline
    ~\cite{Cournier05:snap-stabilizing} & $O(\log n)$ & 0 & Snap Stabilizing  \\ & & & Wave takes $O(n^2)$ rounds \\ \hline
    ~\cite{HuangChen} & $O(\log n)$ & $O(nD)$ & Token Circulation, not silent\\ \hline
    ~\cite{Petit:DFTC} & $O(\log n)$ & $O(n)$ & Token Circulation, not silent  \\ \hline  
    ~\cite{Datta:DFTC} & $O(\log \Delta)$ & $O(nD)$ & Token Circulation, not silent \\ \hline
    ~\cite{Johnen95:DFTC} & $O(\log \Delta)$ & $O(nD)$ & Token Circulation, not silent \\ & & & Requires neighbor of neighbor info \\ \hline
    ~\cite{johnen97:DFTC} &  $O(\Delta)$ & $O(nD)$ & Token Circulation, not silent\\ \hline
    ~\cite{PetitV97:DFTC} & $O(\log \Delta)$ & $O(nD)$ & Token Circulation, not silent \\ \hline
    \textbf{OUR RESULTS} & $O(\log n)$ & $O(n)$ &\textbf{Two algorithms: Silent and}\\
     & & &\textbf{token circulation;}\\& & & \textbf{both with the same complexity} \\ \hline
 \end{tabular}   
 \end{center}
 \caption{Comparing self-stabilizing DFS algorithms}
 \end{table}
Considerable work has been invested in developing self-stabilizing depth-first token circulation algorithms with multiple successive papers improving each other. All of these algorithms also generate a DFS tree in every token circulation round, however these algorithms are not silent. Self-stabilizing depth-first token circulation on arbitrary rooted networks was first considered by Huang and Chen in~\cite{HuangChen}. Their algorithm stabilizes in $O(nD)$ time with a space complexity of $O(\log n)$ bits per node. Subsequently several self-stabilizing DFS token circulation algorithms~\cite{Datta:DFTC,Johnen95:DFTC,johnen97:DFTC,PetitV97:DFTC} were devised. All these papers worked on improving the space complexity of~\cite{HuangChen} from $O(\log n)$ to a function of $\Delta$, the highest degree of a node in the network. The time complexity of all of the above token circulation algorithms~\cite{HuangChen,Johnen95:DFTC,johnen97:DFTC,PetitV97:DFTC} is $O(nD)$ rounds, which is much more than the time it takes for 
one token circulation cycle on a given network. Petit improved the stabilization time complexity of depth-first token circulation to $O(n)$ in~\cite{Petit:DFTC} with a space complexity of $O(\log n)$ bits per node. Petit and Villain~\cite{PetitV00:DFTCM} presented the first self-stabilizing depth-first token circulation algorithm that works in asynchronous message passing systems. 
 
\subsection{Our Contribution}
The main contribution of the current paper is a silent self-stabilizing DFS spanning tree algorithm. The space complexity of our algorithm is $O(\log n)$ bits per node. The only other \emph{silent-stabilizing} DFS tree algorithm~\cite{SelfStabDFS} has a space complexity of $O(n\log \Delta)$. Dolev et al.~\cite{Dolev:lowerbound} established a lower bound of $O(\log n)$ bits per node on the memory requirement of silent-stabilizing spanning tree algorithms. Thus, ours is the first memory optimal silent-stabilizing DFS spanning tree algorithm. 
The silent-stabilizing DFS construction algorithm is designed in a modular way consisting of separate modules for fault detection and correction. The distributed verification module of this algorithm can be considered a contribution in itself.

Composing self-stabilizing primitives using fair combination of protocols is a well-known technique(see e.g.~\cite{DolevIsraeliMoran,Stomp}) to ensure that the resulting protocol is self-stabilizing. We use this approach of protocol combination to design a self-stabilizing depth-first token circulation algorithm which uses our silent-stabilizing DFS tree as a module of the overall algorithm. The space and time complexities of our token circulation algorithm are as good as the previously published work on \emph{fast} self-stabilizing depth-first token circulation~\cite{Petit:DFTC}. 

\subsection{Outline of the paper}
In the next section (Section~\ref{Sec: Definitions}), we describe the model of distributed systems considered in this paper. That section also includes some basic definitions and notations. Section~\ref{Sec:Verification} addresses the distributed verification algorithm which acts as the \emph{Verifier} $\mathcal{V}$ of the proof labeling scheme. The \emph{Marker} $\mathcal{M}$ of the \emph{proof labeling scheme} is presented in Section~\ref{Sec:marker}. Section~\ref{Sec:Stabilization} describes the technique used to make the algorithm self-stabilizing. Section~\ref{Sec:Analysis} presents the correctness proofs and performance analysis. Section~\ref{Sec:circulation} describes a token circulation scheme based on the new silent-stabilizing DFS spanning tree.

\section{Preliminaries}\label{Sec: Definitions}
A distributed system is represented by a connected undirected graph $G(V,E)$ without self-loops and parallel edges, where each node $v\in V$ represents a processor in the network and each edge $e \in E$ corresponds to a communication link between its incident nodes. Processors communicate by writing into their own shared registers and reading from the shared registers of the neighboring processors. The network is assumed to be \emph{asynchronous}. We do not require processors to have unique identifiers. We do assume the existence of a distinguished processor, called the root of the network. Each node $v\in V$ orders its edges by some arbitrary ordering $\alpha_v$ as in~\cite{SelfStabDFS}. For an edge $(u, v)$, let $\alpha_u(v)$ denote the index of the edge $(u,v)$ in $\alpha_u$.

As opposed to Collin and Dolev~\cite{SelfStabDFS}, We use the (rather common) ideal time complexity which assumes that a node reads all of its neighbors in
at most one time unit. Our results translate easily to an alternative, stricter, \emph{contention} time complexity used by Collin and Dolev in~\cite{SelfStabDFS}, where a node can access only one neighbor in one time unit. The time cost of such a translation is a multiplicative factor of $O(\Delta)$, the maximum degree of a node (it is not assumed that $\Delta$ is known to nodes). As is commonly assumed in the case of self-stabilization, each node has only some bounded number of memory bits available to be used. Here, this amount of memory is $O(\log n)$.

\noindent\textbf{Self-stabilization and silent-stabilization:}
A distributed algorithm is self-stabilizing if it can be started in any arbitrary global state and once started, the algorithm converges to a legal state by itself and stays in the legal state unless additional faults occur~\cite{Dolev}. A self-stabilizing algorithm is \emph{silent} if starting from an arbitrary state it converges to a legal global state after which the values stored in the communication registers do not change, see e.g.~\cite{Dolev:lowerbound}. While some problems like token circulation are non-silent by nature, many \emph{input/output} algorithms allow a silent solution.

\noindent\textbf{Spanning Tree: Distributed Representation:}
A spanning tree $T$ of a connected, undirected graph $G(V,E)$ is a tree composed of all the nodes and some of the edges of $G$. A spanning tree $T$ of some graph $G$ is represented in a distributed manner by having each node locally mark some of its incident edges such that the collection of marked edges of all the nodes forms a spanning tree of $G$. Actually, it is enough that each node marks its edge leading to its parent on the tree in a local variable.

\noindent\textbf{DFS Tree and the \emph{first} DFS Tree of a Graph:}
A DFS Tree of a connected, undirected graph $G(V,E)$ is the spanning tree generated by a depth first search traversal of $G$. In a DFS traversal, starting from a specified node called the root, all the nodes of the graph are visited one at a time, exploring as far as possible before backtracking, see e.g.~\cite{Even}. 
The \emph{first} DFS traversal is the one that acts as follows:
whenever a node $v$ has a set of unexplored edges to choose from, the chosen edge is the edge with the smallest port number in the port ordering $\alpha_v$. The tree thus generated is called the \emph{first DFS} tree~\cite{SelfStabDFS}. While a connected, undirected graph can have more than one DFS spanning trees, it can have only one \emph{first DFS} spanning tree.

\noindent\textbf{Lexicographic Ordering}\label{Sec: LexicOrdering}
A simple path from the root of a graph $G$ to some node $v \in V$ can be represented as a string starting with a $\perp$ followed by a sequence of the port numbers of the outgoing edges on the path~\cite{SelfStabDFS}. Given such a string representation of a path, a lexicographic operator $\prec$ can be defined to compare multiple paths of a given node $v$ from the root, where $\perp$ is considered the minimum character. In the \emph{first} DFS tree of a graph, the path leading from the root to some node $v \in V$ is the lexicographically smallest (w.r.t. $\prec$)  among all the simple paths from the root to $v$~\cite{SelfStabDFS}. 

\noindent\textbf{DFS Intervals}
In a DFS traversal, it is common to assign to each node an interval $(in, out)$ corresponding to the discovery and finish time of exploration of that node. The discovery time or $in$ is the time at which a node is discovered for the first time. The discovery time of a node $v \in V$ is denoted as $in_v$. The finish time of node $v$ denoted by $out_v$ is the time at which a node has finished exploring all its neighbors. These intervals have the property that given any two intervals $(in, out)$ and $(in', out')$, either one includes the other or they are totally disjoint. Assuming without loss of generality that $in < in'$, we can write this formally as: either $(in < in' < out' < out)$ or $(in < out < in' < out')$~\cite{Even}. In other words, the DFS intervals induce a partial order on the nodes of a graph. 

\subsection{Notation}\label{Sec: Notations}
We define the following notation to be used throughout:
\begin{itemize} 
\item $\eta(v)$ denotes the set of neighbors of $v$ in $G$. $\forall v \in V$ $(\eta(v)$ $=$ $ \{u | u \in V \wedge (u,v) \in E)\})$. 
\item $interval_v$ denotes the $(in, out)$ \emph{label} of $v$.
 \item $in_v$ denotes the $in$ \emph{label} of $v$ and $out_v$ denotes the $out$ \emph{label} of $v$.
\item Relational operator $\subset$ between two intervals $(in, out)$ and $(in', out')$ indicates the inclusion of of the first interval in the second one. For example: $(in, out) \subset (in', out')$ indicates that $(in, out)$ is included in $(in', out')$.
\item Relational operator $\supset$ is defined similarly.
\end{itemize}

\section{DFS Verification: \emph{Verifier} $\mathcal{V}$}\label{Sec:Verification}
Given a graph $G(V,E)$ and the distributed representation of a spanning tree $T$ of $G$, the DFS verification algorithm is required to verify that $T$ is the \emph{first} DFS tree of $G$. 
The \emph{Verifier} $\mathcal{V}$ takes as input a connected graph $G(V,E)$ where each node $v \in V$ bears an $(in_v, out_v)$ label in addition to $v$'s parent on $T$. Note that $\mathcal{V}$ takes $(in, out)$ labels of nodes as input and is not concerned with how they are generated. \\
We assume that each node can read the labels of all its neighbors in addition to its own label and state. A node cannot look at the state of any of its neighbors, however. Each node $v \in V$ periodically reads the labels of all its neighbors and locally computes the following additional information from its own state and label as well as the labels of its neighbors. 
\subsection{Intermediate Computations} \label{sec: data}
Each node computes the following \emph{macros} to be used for verification.
\begin{enumerate}
\item There are zero or more neighbors of $v$ whose interval includes $v$'s interval. Let us call the set of all such nodes the \emph{neighboring ancestors} of $v$ and denote this set by by $anc_l(v)$.
 $$anc_l(v) = \{w | w \in \eta(v) \; and  \; interval_w \supset interval_v\}$$
 \item The parent of $v$ as perceived by the labels : $parent_l(v)$ $=$ $w | w \in anc_l(v) \wedge \forall u \in anc_l(v)$ ($u \neq w \rightarrow $ $interval_w$ $\subset$ $interval_u$).  
\item There are zero or more neighbors of $v$ whose interval is included in $v$'s interval, let us call the set of all such nodes the \emph{neighboring descendants} of $v$ and denote this set by $desc_l(v)$.
$$desc_l(v) = \{ w | w \in \eta(v)\; and \; interval_w \subset interval_v\}$$
\item A \emph{child neighbor} of $v$ is a neighboring descendant of $v$ whose interval is not included in the interval of any other neighboring descendant of $v$. 
  $$child_l(v) = u | u \in desc_l(v) \wedge \neg\exists u' \in desc_l(v) (u' \neq u \wedge interval_{u'} \supset interval_u)$$
\item $children_l(v)$ $\subseteq$ $desc_l(v)$ is the set of all \emph{child neighbors} of $v$. 
\end{enumerate}
The subscript $l$ in $anc_l(v)$ above denotes that the set $anc_l(v)$ is computed by the node $v$ by just looking at the labels of $v$ and those of $v$'s neighbors. The same holds for all the other \emph{macros} defined above. 
It is worth pointing out that all these are intermediate computations and the data they generate need not be stored on the node. 

The verification is performed by having each node compute a set of predicates. If $T$ is indeed the \emph{first} DFS tree of $G$ and the labels on all the nodes are proper (i.e. they are as if they were generated by an actual \emph{first} DFS Traversal of the input graph); then the verifier \textbf{\emph{accepts}} continuously on every node until a fault occurs. If a fault occurs either due to the corruption of the state of some nodes or due to some nodes having incorrect labels, at least one node \textbf{\emph{rejects}}. The node that rejects is called a detecting node. The verifier self-stabilizes trivially since it runs periodically. 

\subsection{Local Interval Predicates} \label{Sec: predicates}
Let $parent_v$ denote the local variable used to store the parent of $v$ in $T$. Following is the set of local predicates that each node has to compute:
 \subsubsection{Predicates for the root node $r$}\label{Sec: predicates-root}
\begin{enumerate}
\item $parent_r =  null$.
\item $anc_l(r) = \phi$.
\end{enumerate}

\subsubsection{Predicates for a non-root node $v$}\label{Sec: predicates-nonroot}
\begin{enumerate}
\item \label{predicate: nonNull} $parent_v \neq null$.
\item \label{predicate: nonEmpty} $anc_l(v) \neq \phi$.
 \item \label{predicate: sameTrees} $parent_v = parent_l(v)$. The parent of $v$ on $T$ denoted by $parent_v$ is the same as $v$'s parent as computed by $v$ from the labels of $v$ and its neighbors.
\item \label{predicate: childParent} $interval_v$ $\subset$ $interval_{parent_v}$.
\item \label{predicate: childAnc} $\forall u \in anc_l(v)$ such that $u \neq parent_v$ ($interval_{parent_v}$ $\subset$ $interval_u$).
\end{enumerate}

\subsubsection{Predicates for every node(root as well as a non-root) $v$}\label{Sec: predicates-all}
\begin{enumerate}
\item \label{predicate: sanity} $out_v > in_v$.
\item \label{predicate: disjoint} There is no neighbor of $v$ such that its interval is totally disjoint with $v$. Formally \\ $\forall u \in \eta(v)$ ($interval_u \subset interval_v \vee interval_u \supset interval_v$).  
\item \label{predicate: leaf} if $|children_l(v)| = 0 $  then $out_v = in_v + 1$.
\item \label{predicate: firstLastChild} if $|children_l(v)| > 0$ and let $childrenD_l(v)$ denote the list of children of $v$ sorted in ascending order of their $in$ labels and $firstChild_l(v)$ and $lastChild_l(v)$ be the first and last members of $childrenD_l(v)$ then $in_{firstChild_l} = in_v +1$ $\wedge$ $out_v = out_{lastChild_l} + 1$.  
\item \label{predicate: childrenOrder} if $|children_l(v)| > 1$ and let $childrenP_l(v)$ denote the list of children of $v$ sorted in the ascending order of their port numbers in $v$, then $childrenD_l(v)$ and $childrenP_l(v)$ sort the members of $children_l(v)$ in the same order.
\item \label{predicate: childDesc} Let $u$ and  $w \in desc_l(v)$, $u \neq w$, such that $u \in children_l(v)$ and $w \notin children_l(v)$ and $in_u < in_w$ then $\alpha_v(u) < \alpha_v(w)$.
\item\label{predicate: inOut} $\forall (u,w) \in childrenD_l(v)$ such that $u$ and $w$ are adjacent in $childrenD_l(v)$ and $in_u < in_w$, then $in_w = out_u + 1$ 
\end{enumerate}

\begin{remark}
 The only predicates that deal with the order in which the neighbors of a node are explored are~\ref{predicate: childrenOrder} and~\ref{predicate: childDesc} of Section~\ref{Sec: predicates-all}. Omitting these two Predicates leaves us with a set of predicates sufficient to verify that $T$ is \emph{some} DFS tree(may not be same as the initial input to the verifier) of $G$. If an algorithm that uses the verifier as a subtask is not concerned about the order, it can simply drop these predicates.
 \end{remark}


\section{Generating the Labels: \emph{Marker} $\mathcal{M}$}\label{Sec:marker}
A natural method for assigning the $(in, out)$ labels is to perform an actual DFS traversal of the network starting from the root. The required labels can be generated by augmenting some known DFS tree construction algorithm (e.g.~\cite{chlamtac1987tree},~\cite{DFSConstruction1},~\cite{DFSConstruction2}) by adding new variables for the labels and specific actions for updating these label variables. We assume that the DFS construction algorithm of Awerbuch~\cite{DFSConstruction1} can be easily translated to shared memory and the resulting algorithm can be easily augmented with actions to update the $in$ and $out$ labels. Note that translating~\cite{DFSConstruction1} to shared memory is trivial and it decreases the memory from $O(\Delta)$ to $O(log\Delta)$, if it changes memory at all, since a node does not need to store the \emph{VISITED} message(the message broadcasted by a node to all its neighbors when it is visited for the first time, See~\cite{DFSConstruction1}) of a neighbor, instead it can read the 
shared register of the neighbor.The pseudo code of the marker will appear in the full paper.

\section{The Silent-Stabilizing DFS Construction Algorithm}\label{Sec:Stabilization}
We have constructed a \emph{proof labeling scheme} $(\mathcal{M}, \mathcal{V})$ with a non-stabilizing marker $\mathcal{M}$ that takes as input a connected graph $G$ and assigns $(in, out)$ labels to every node in $G$. It also has a verifier $\mathcal{V}$ that takes as input a labeled (with $(in,out)$ intervals) distributed data structure and verifies whether the input structure is the \emph{first} DFS tree. The proofs for the correctness and the performance of $(\mathcal{M}, \mathcal{V})$ are presented in Section~\ref{Sec:Analysis}. In the meanwhile, we use them here assuming they are correct.

A simple way to stabilize any input/output algorithm is to run the algorithm repeatedly to maintain the correct output along with a self-stabilizing synchronizer~\cite{AwerbuchTransformer}. This however would not be a silent algorithm. Still, let us use this approach to generate a non-silent self-stabilizing algorithm as an exercise, before presenting the silent one. Awerbuch and Varghese, in their seminal paper~\cite{AwerbuchTransformer}, present a transformer algorithm for converting a non-stabilizing input/output algorithm into its self-stabilizing version.  
Following theorem is taken from the paper of Awerbuch and Varghese~\cite{AwerbuchTransformer}:
\begin{theorem}\label{AwerbuchVargheseTheorem}
Given a non-stabilizing distributed algorithm $\Pi$ to compute an input/output relation with a space complexity of $S_\Pi$ and a time complexity of $T_\Pi$. The Resynchronizer compiler produces a self-stabilizing version of $\Pi$ whose time complexity is $O(T_\Pi + \hat{D})$ and whose space complexity is same as that of $\Pi$, where $\hat{D}$ is an upper bound on the diameter of the network.
\end{theorem}
 Informally, the transformer that Awerbuch and Varghese developed to prove the above theorem is a self-stabilizing synchronizer. The transformer takes as input a non-stabilizing input/output algorithm $\Pi$ whose running time and space requirement are $T_\Pi$ and $S_\Pi$ respectively. Another input it takes is $\hat{D}$ which is an \emph{upper bound} on the actual diameter $D$ of the network. Given these inputs, the transformer performs $\Pi$ for $T_\Pi$(recall that the transformer is a synchronizer and transforms the network to be synchronous). Then it retains the results, performs $\Pi$ again and compares the new results to the old ones. If they are the same, the old results are retained. if they differ, then some faults occurred, the new results are retained. This is repeated forever.

Since we do not assume the knowledge of $n$ (required for input : $T_\mathcal{M}$) or $\hat{D}$, we use a slightly modified version of theorem~\ref{AwerbuchVargheseTheorem} here, that appeared in~\cite{KormanKuttenSSMST}. The modified Awerbuch Varghese theorem presented in~\cite{KormanKuttenSSMST} is as follows:
\begin{theorem}\label{KuttenKormanTheorem}
Given a non-stabilizing distributed algorithm $\Pi$ to compute an input/output relation with a space complexity of $S_\Pi$ and a time complexity of $T_\Pi$. The enhanced Resynchronizer compiler produces a self-stabilizing version of $\Pi$ whose time complexity is $O(T_\Pi + n)$ for asynchronous networks and $O(T_\Pi + D)$ for synchronous networks with a space complexity of $O(S_\Pi + logn)$.
\end{theorem}
Informally, Korman et al. used a better synchronizer plus a simple self-stabilizing algorithm that computes $n$ and $D$ to prove the above theorem.
To obtain a non-silent self-stabilizing DFS construction algorithm, we just plug the marker $\mathcal{M}$ of Section~\ref{Sec:marker} into theorem~\ref{KuttenKormanTheorem} and obtain the following corollary.
\begin{corollary} \label{corollary1}
 There exists a non-silent self-stabilizing DFS construction algorithm that can operate in a dynamic asynchronous network, with a time complexity of $O(T_\mathcal{M} + n)$ and a space complexity of $O(S_\mathcal{M} + \log n)$.
 \end{corollary}

\subsection{Achieving Silent-Stabilization}\label{Sec:SilentStab}
Before going into the details of achieving silence, let us go over how the self-stabilizing synchronizer of the enhanced transformer of theorem~\ref{KuttenKormanTheorem} helps co-ordinate repeated executions of the marker in the algorithm of corollary~\ref{corollary1}. A synchronizer simulates a synchronous protocol in an asynchronous network by using a pulse count at each node which is updated in increments of $1$ subject to certain rules. A node $u$ executes the $i$th step of the algorithm when pulse count at $u$, $pulse_u$ is equal to $i$. The synchronizer maintains the \emph{invariant} that the pulse count of a node $u$ differs from any of its neighbors by at most one. Since the synchronizer module is self-stabilizing, all the nodes may be initialized to an arbitrary pulse count and thus the network may not be synchronized in the beginning. The stabilization time of the synchronizer module of the enhanced transformer is $O(n)$, thus starting from any arbitrary set of pulse counters, the network is 
guaranteed to be synchronized after $O(n)$ time.
The enhanced transformer waits for \emph{sufficient} time for the nodes to get synchronized and then starts the execution of the algorithm to be stabilized, in our case, the marker $\mathcal{M}$. If $T_e$ denotes the pulse count at which all the nodes are synchronized, the nodes run the marker from $T_e$ to $T_e + T_\mathcal{M}$. Due to an allowed difference of at most $1$ between pulse counts of neighboring nodes, the maximum difference between the pulse counts of any two nodes is $D$, the diameter of the network. Thus any node with a pulse count of $T_e + T_\mathcal{M}$ has to wait a maximum of $D$ pulses to be sure that all the nodes in the network have written their output~\cite{AwerbuchTransformer}. The node with a pulse count of $T_e + T_\mathcal{M} + D$ wraps around its pulse count to $0$ which destroys the synchronization. Essentially the first node(s) to \emph{wrap around} invoke the \emph{reset} module of the transformer which brings the nodes back in sync for the next execution of the marker.
To make the algorithm silent-stabilizing, we execute the marker(along with the synchronizer) only once in the beginning to generate the labels. The silence is achieved by turning the synchronizer off after all the nodes have finished executing the marker. As explained above, the nodes can easily detect when the marker has finished by looking at their respective pulse counts. When a node reaches a pulse count of $T_e + T_\mathcal{M} + D$, it stops updating its pulse count, thus turning the synchronizer off. When all the nodes in the neighborhood of a node have reached $T_e + T_\mathcal{M} + D$, it turns on the verifier $\mathcal{V}$. Since $\mathcal{V}$ can detect a fault in exactly one pulse, if one occurs, we can manage without running a synchronizer during the verification. The verifier keeps running repeatedly until a fault occurs. If a node $v$ detects a fault, it invokes the synchronizer of the enhanced transformer again by dropping $v$'s pulse count to $0$. Again, as 
in case of non-silent algorithm, this invokes a reset which resynchronizes the network and subsequently invokes the marker again.
Note that the nodes need not know the $T_\mathcal{M}$ a priori. The running time of $\mathcal{M}$ is a function of $n$, the number of nodes which can be computed in a self-stabilizing manner by the module of the enhanced transformer responsible for computing $n$.
 
\begin{observation}
 The only communication that takes place at each node during verification is the reading of the shared registers of the neighbors. 
 The computations performed during verification do not affect the contents of the shared registers at all, thus ensuring silence as defined in~\cite{Dolev:lowerbound}.
\end{observation}

Thus we obtain a silent-stabilizing DFS construction algorithm.
 The following theorem summarizes our result:
\begin{theorem}\label{thrm: TransformerMV}
 The \emph{proof labeling scheme} $(\mathcal{M}, \mathcal{V})$ for a DFS tree implies a silent-stabilizing DFS construction algorithm, that runs in $O(T_\mathcal{M} + n)$ time with a space complexity of $O(S_\mathcal{M} + S_\mathcal{V} + \log n)$.
\end{theorem}
 
\section{Correctness and Performance Analysis}\label{Sec:Analysis}
In this section, we establish the correctness of our algorithm. The proofs follow easily from the known properties of a DFS tree and the predicates of the verifier. \\
Given a labeled (with $(in, out)$ labels) graph $G(V,E)$ and the distributed representation of a spanning subgraph $T$ of $G$, the following lemmas holds on $G$, if the local interval predicates (Section~\ref{Sec: predicates}) hold true at every node of $G$: 

\begin{lemma}\label{lemma:tree}
$T$ is a spanning tree of $G$.
\end{lemma}
\begin{proof}
In order to prove that a graph is a tree, it is sufficient to prove that it has no cycles and its number of edges is $n-1$, where $n$ is the number of nodes in this graph~\cite{Even}. For the subgraph $T$ of $G$ to have a cycle, one of the ancestors of some node $v \in V$ has to mark $v$ as its parent. However, this leads to a contradiction by predicate~\ref{predicate: childParent} of Section:~\ref{Sec: predicates-nonroot} which requires that the interval of a node be included in the interval of its parent. Applying predicate~\ref{predicate: childParent} to $v$ and $v$'s ancestors, implies that for an ancestor $u$ of $v$ which points to $v$ as its parent, $interval(v) \subset interval(u) \wedge interval(u) \wedge interval(v)$, a contradiction.
The parent pointer of each node $v \in V$ except the root comprises of a single incident edge of $v$ and the parent pointer of the root is $null$, therefore there are exactly $n$ nodes and $n-1$ edges in $T$. 
\end{proof}

\begin{observation}\label{observation: sameTree}
The \emph{macros} defined in Section~\ref{sec: data} extract (periodically) a perceived tree $T_l$ from the $(in, out)$ labels of the nodes in $G$.  
\end{observation}
While input tree $T$ is encoded only by the collection of the parent pointers of the nodes, $T_l$ is extracted by having each node compute its perceived parent, denoted by $parent_l$ as well as its perceived children, denoted by the set $children_l$ on $T_l$.
\begin{lemma}\label{lemma: sameTree}
 For any node $v \in V$, the set of children of $v$ in $T$ is same as the set of perceived children of $v$ in $T_l$. 
\end{lemma}
\begin{proof}
The predicate~\ref{predicate: sameTrees} of section~\ref{Sec: predicates-nonroot}, ensures that the parent pointer $parent_v$ of a node $v$ on the input tree $T$ is the same as $v$'s perceived parent $parent_l(v)$ on $T_l$. The set of children of a node $v$ on $T$ is implicitly implied by the parent pointers of $v$'s children. Hence, it is sufficient to prove that the set of perceived children of $v$ on $T_l$ is the same as those implied by the perceived parent pointers of perceived children of $v$, i.e., the collection of perceived parents is consistent with the collection of perceived children on $T_l$.
In what follows, we prove that if a node $v$ has a node $p$ as its perceived parent ($parent_l(v) = p$), then $v \in children_l(p)$. Assume, for contradiction, that the above does not hold. Note that, by the definition of a perceived parent and simple inductive arguments, $p$ has the \emph{narrowest} interval of any node whose interval includes $interval_v$, i.e., the interval of $p$ does not include the interval of any other node whose interval includes $interval_v$. 
Having $v \notin children_l(p) \wedge parent_l(v) = p$ implies that there is a node $x \in \eta(p)$ with $interval_x \supset interval_v$ and moreover $interval_p \supset interval_x$. This implies that $p$ can not be the parent of $v$.
In a similar way, one can prove that if $c \in children_l(v)$ then $v$ is the perceived parent of $c$. 
\end{proof}

Following lemma~\ref{lemma: sameTree}, in the discussion that follows, $children_l(v)$ implies the children of $v$ in $T$ and vice versa. 
\begin{lemma} \label{lemma: disjoint}
For any two children $u, w$ of a node $v$ in $T$, the intervals of all the nodes in the subtree of $u$ in $T$ are disjoint from the intervals of all the nodes in the subtree of $w$ in $T$.
\end{lemma}
\begin{proof}

 The set $childrenD_l(v)$ is the set $children_l(v)$ sorted in the ascending order of the $in$ labels of the nodes $\in children_l(v)$ as defined in Section~\ref{Sec: predicates-all}. Let us assume, without loss of generality, that $in_w > in_u$. Consider a node $u' \in \eta(v)$ such that $u'$ is adjacent to $u$ and appears after $u$ in $childrenD_l(v)$(possibly $u'=w$). Applying predicate~\ref{predicate: inOut} of Section~\ref{Sec: predicates-all} to $u$ and $u'$ , $in_{u'} = out_u + 1$. By predicate~\ref{predicate: sanity} of Section~\ref{Sec: predicates-all}, $out_{u'} > in_{u'}$. Thus neither of the two intervals, $interval(u)$ and $interval(u')$, includes the other, i.e. they are totally disjoint.
Applying predicate~\ref{predicate: childParent} of Section:~\ref{Sec: predicates-nonroot} inductively, it is easy to see that the intervals of all the descendants of $u$ in $T$ are included in $u$'s own interval. Similarly, the intervals of all the descendants of $u'$ are included in $u'$'s interval . Therefore, intervals of all the descendants of $u$ are disjoint from the intervals of $u'$ and all its descendants. By inductively applying the above argument to every adjacent pair of nodes in $childrenD_l(v)$ starting from $u'$ to $w$, it is easy to show that the subtrees of any two children of a node have disjoint intervals.
\end{proof}


\begin{lemma}\label{lemma: cross}
For any two children $u, w$ of some node $v$ in $T$, every simple path in $G$ from some node in the subtree of $u$ to any node in the subtree of $w$ in $T$ goes through either $v$ or $v$'s ancestors. 
\end{lemma}
\begin{proof}

\begin{figure}
[ht!]
\centering 
\caption{Figure for proof of lemma~\ref{lemma: cross}}
\subfigure[Case 1:path through a descendant of a sibling of $u$.]{\label{sfig: lemma3Case1} \includegraphics[scale=0.35]{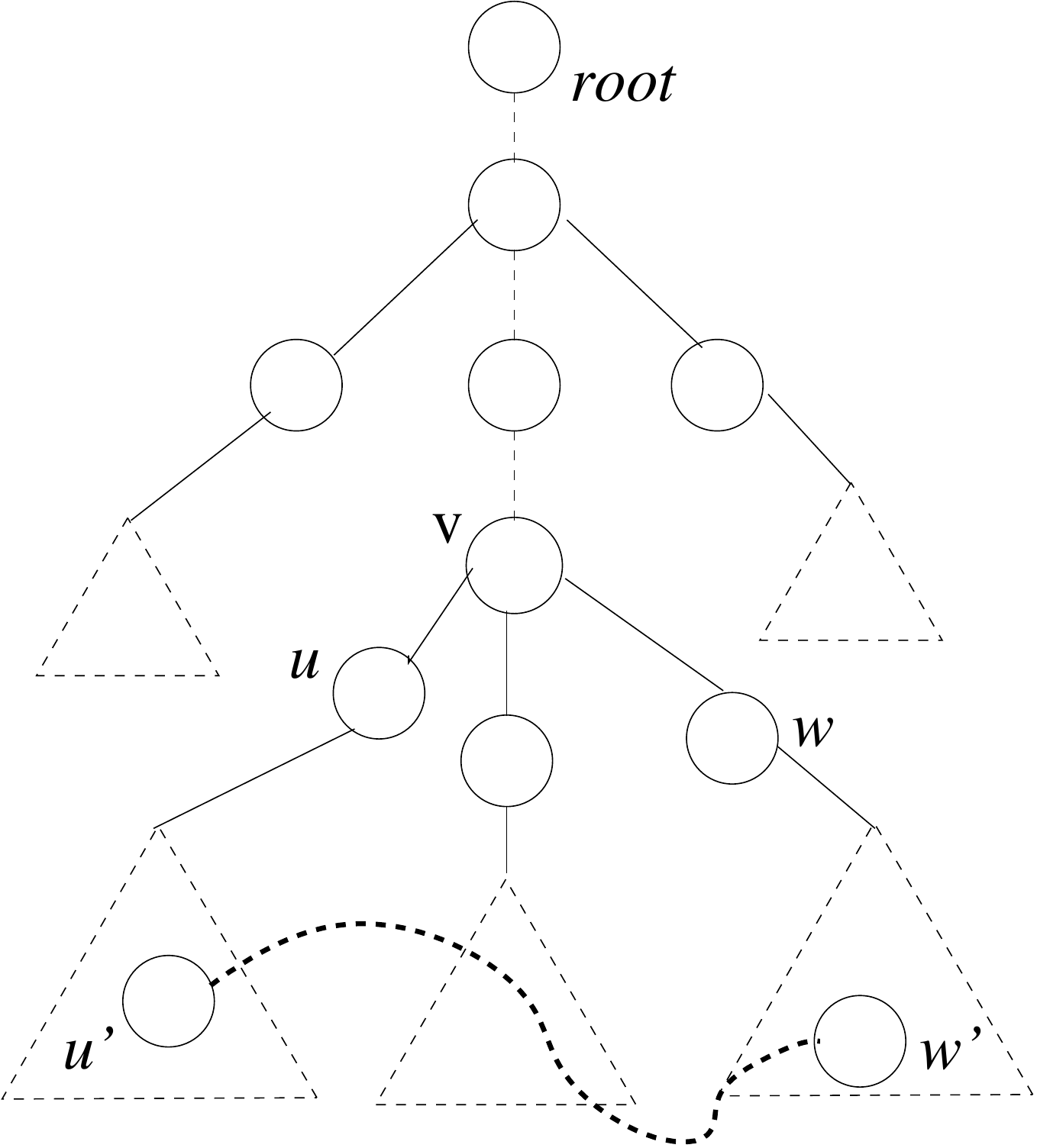}} \hspace{15pt}
\subfigure[Case 2: path through a descendant of a sibling of an ancestor of $u$ and $w$]{\label{sfig: lemma3Case2} \includegraphics[scale=0.35]{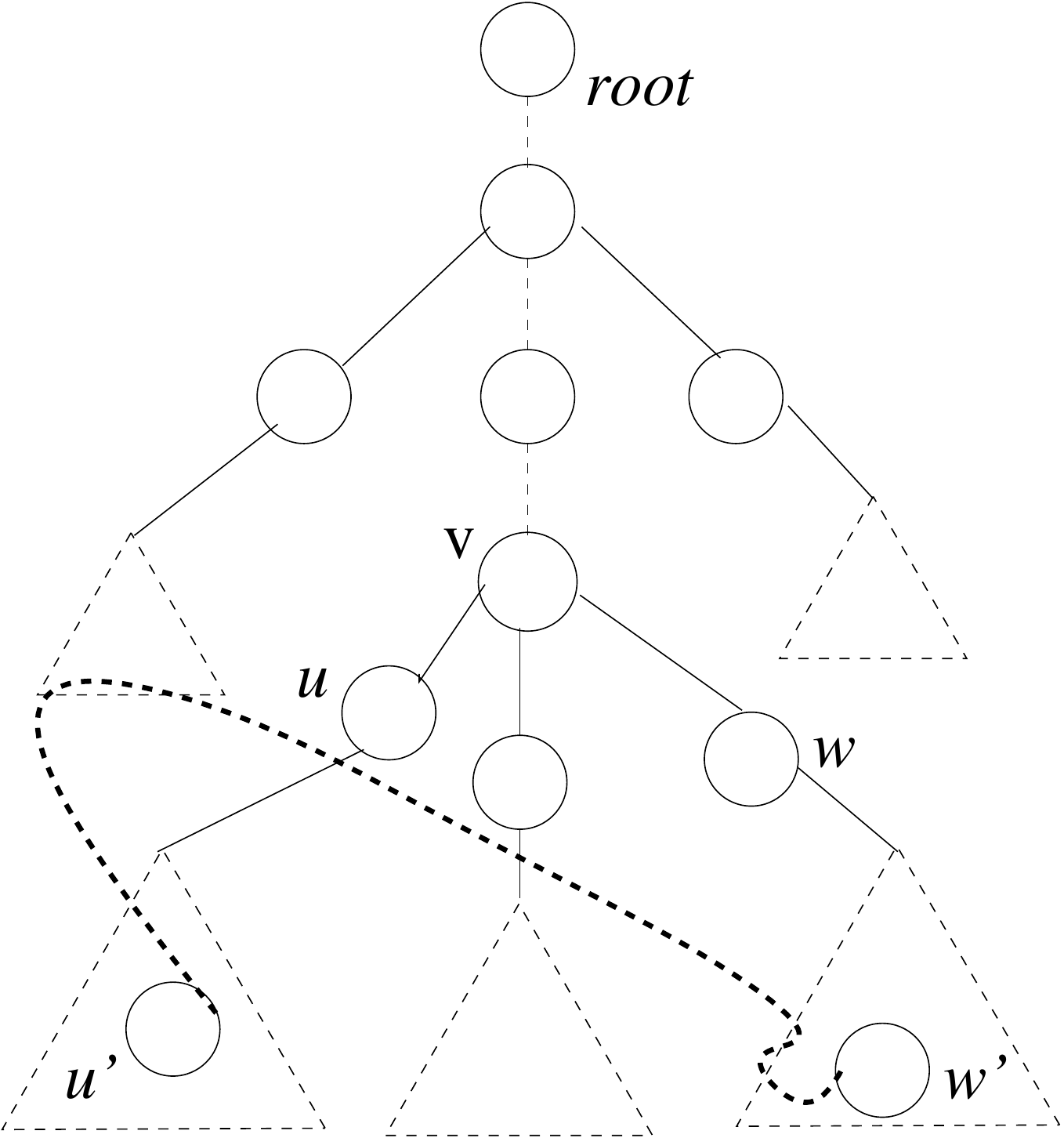}}
\end{figure}
Let $u'$ be some node in the subtree of $u$ and $w'$ be some node in the subtree of $w$. Let us assume, by way of contradiction, that there is a simple path $P$ in $G$ between $u'$ and $w'$ that does not go through $v$ or $v$'s ancestors. There are two possibilities:
\begin{itemize}
 \item $P$ goes through a descendant of a sibling of $u$ (possibly $w$). 
 \item or, it goes through a descendant of a sibling of an ancestor (possibly $v$) of $u$ and $w$.
 \end{itemize}
Both these cases require an edge to exist in $G$ that connects a pair of nodes in two sibling subtrees, known as a \emph{cross} edge~\cite{Even}. By lemma~\ref{lemma: disjoint}, the intervals of all the nodes in the subtree of some node $x$ are disjoint from the intervals of all the nodes in the subtree of a sibling of $x$. Thus, the existence of any such edge in $G$ is ruled out by predicate~\ref{predicate: disjoint} of Section~\ref{Sec: predicates-all}. 
\end{proof}

\begin{observation}
 The proof of Lemma~\ref{lemma: cross} shows that there are no cross edges in the input tree $T$ which implies that $T$ is a DFS(not necessarily the \emph{first} DFS) tree of $G$.
 \end{observation}

\begin{theorem}\label{thrm: verifierC}
 If a graph $G(V, E)$ has every node $v \in V$ labeled with its $(in, out)$ interval and interval assignments are such that all the local interval predicates (Section~\ref{Sec: predicates}) hold true at every node, then the spanning tree $T$ encoded in a distributed manner in the states of all the nodes of $G$ is the \emph{first} DFS tree of $G$. 
\end{theorem}

\begin{proof}
The problem of finding the \emph{first} DFS Tree of a graph can be thought of as the one of selecting the lexicographically smallest simple path of every node $v \in V$ out of all the simple paths from the root to $v$, see~\cite{SelfStabDFS}. Let $P^{T}_v$ denote the path leading from the root to some node $v$ in $T$. We now prove that for any node $v \in V$, $P^{T}_v$ is the lexicographically smallest among all the simple paths from the root to $v$ in $G$. By way of contradiction, let us assume that there is another simple path $P^{Alt}_v$ from the root to $v$ which is smaller than $P^{T}_v$. Let us assume, $w.l.o.g.$, that $P^{T}_v$ and $P^{Alt}_v$ are the same up-to(and including) some node $v_m$, the $m^{th}$ node of the common prefix. Let $v^{T}_{m+1}$ and $v^{Alt}_{m+1}$ denote the $(m + 1)^{th}$ node of $P^{T}_v$ and $P^{Alt}_v$ respectively.

\begin{observation}
For $P^{Alt}_v$ to be lexicographically smaller than $P^{T}_v$, the edge index (as defined in Section~\ref{Sec: LexicOrdering}) $\alpha_{v_m}(v^{Alt}_{m+1})$ must be smaller than the corresponding index $\alpha_{v_m}(v^T_{m+1})$.
\end{observation}

 There are three possibilities for $P^{Alt}_v$ based on how $v^{Alt}_{m+1}$ is related to $v_m$ :
 \begin{enumerate}
\item \emph{$v^{Alt}_{m+1}$ is an ancestor of $v_m$}\label{ancestor}: This case is ruled out since any such path will not be a simple path. 
\item \emph{$v^{Alt}_{m+1}$ is a child of $v_m$}\label{disjointChildren}: $v^{Alt}_{m+1}$ and $v^{T}_{m+1}$ are both children of $v_m$. According to lemma~\ref{lemma: cross}, there is no simple path from $v^{Alt}_{m+1}$ to any node in the subtree of $v^T_{m+1}$ that does not go through $v_m$ or any of its ancestors. Since $v^{T}_{m+1}$ falls on $P^{T}_v$, $v$ belongs to the subtree of $v^{T}_{m+1}$ in $T$. Thus, there is no simple path connecting $v^{Alt}_{m+1}$ to $v$ that does not go through $v_m$ or its ancestors. The path from $v^{Alt}_{m+1}$ to $v$ that goes through either $v_m$ or any of its ancestors would not be a simple path as in case~\ref{ancestor}. Therefore, this case is also ruled out.


\item \emph{$v^{Alt}_{m+1}$ is a descendant which is not a child of $v_m$}: This case can be further subdivided into two sub cases:
\begin{enumerate}
\item \emph{$v^{Alt}_{m+1}$ is also a descendant of $V^T_{m+1}$ in addition to being a descendant of $v_m$}: This implies that $in_{v^{alt}_{m+1}} > in_{v^{T}_{m+1}}$. Also, $v^{T}_{m+1}$ is a child of $v_m$. This leads to a contradiction due to local interval predicate~\ref{predicate: childDesc} (Section~\ref{Sec: predicates-all}) which requires that the edge index of the edge $(v_m, v^{T}_{m+1})$  be smaller than the edge index of the edge $(v_m, v^{Alt}_{m+1})$ in $alpha_{v_m}$.

\item \emph{$v^{Alt}_{m+1}$ is a proper descendant of of $v_m$, but not a descendant of $v^T_{m+1}$} : This case is similar to that of~\ref{disjointChildren}. 
\end{enumerate}
 
\end{enumerate}
\end{proof}

\begin{theorem}\label{thrm: verifierP}
The verifier $\mathcal{V}$ described in section~\ref{Sec:Verification} runs in one time unit and requires $O(\log n)$ bits of memory per node.
\end{theorem}

\begin{proof}
The running time of $\mathcal{V}$ follows from the fact that each node needs to look only at the labels of its immediate neighbors in order to compute its predicates. Every node shares its $(in, out)$ labels with its neighbors. The maximum value of a label is $2n$ which can be encoded using $O(\log n)$ bits. 
\end{proof}

The following theorem establishes the correctness and performance of the marker $\mathcal{M}$:
\begin{theorem}
There exists a marker that constructs the \emph{first} DFS tree and assigns $(in, out)$ labels to all the nodes of the input graph $G(V,E)$ in time $O(n)$ using $O(\log n)$ bits of memory per node.  
\end{theorem}
 
 \begin{proof}
 As described in Section~\ref{Sec:marker}, it is easy to design a marker that adds new actions\footnote{Actually, these are just common actions of various versions of non-distributed DFS.} to a standard DFS tree construction algorithm for computing the $in$ and $out$ labels. The standard DFS tree construction algorithm in shared memory model, without any actions for computing the $(in, out)$ labels has a space complexity of $O(\log \Delta)$ bits per node. The variables for updating the $(in, out)$ labels require $O(\log n)$ bits per node. Therefore the overall space complexity of such a marker is $O(logn)$. \\
 The actions for computing the labels do not change the values of any of the variables of the original algorithm. Also, these actions do not change the algorithm's flow of control. The addition of these actions cannot violate the correctness of the construction algorithm, nor change its time complexity of $O(n)$. \\
 It is easy to modify the algorithm such that a node $v$ always picks the unvisited neighbor with the smallest port number. This ensures that the output of the algorithm is the \emph{first} DFS tree of the input graph. 
 \end{proof}
 
 
\section{Self-stabilizing DFS token circulation}\label{Sec:circulation}
The silent-stabilizing DFS tree of Section~\ref{Sec:SilentStab} can be combined with a self-stabilizing mutual exclusion algorithm for tree networks to obtain a self-stabilizing token circulation scheme for general networks with a specified root. Self-stabilizing mutual exclusion algorithms that circulate a token in the DFS order on a tree network can be found in~\cite{Dolev:1993,Petit97highlyspace-efficient,Petit97SelfStabTree}. 
Petit and Villain presented a space optimal \emph{snap-stabilizing} DFS token circulation algorithm for tree networks in~\cite{Petit2007:SnapStabOnTree} with a \emph{waiting time}(See~\cite{Petit2007:SnapStabOnTree} for a definition of waiting time) of $O(n)$. We can combine our silent-stabilizing DFS tree with the snap stabilizing DFS token circulation protocol of~\cite{Petit2007:SnapStabOnTree} using the fair composition method~\cite{DolevIsraeliMoran} to obtain a DFS token circulation for general networks. 
The space complexity of~\cite{Petit2007:SnapStabOnTree} is $O(\log \Delta)$ and that of our silent-stabilizing DFS tree is $O(\log n)$. Therefore the space complexity of the resulting self-stabilizing DFS token circulation algorithm is $O(\log n)$. \\ 

\subsubsection*{Acknowledgements}
This research was supported in part by a grant from ISF and Technion TASP center.


\begin{thebibliography}{10}

\bibitem{Afek_thelocal}
Yehuda Afek, Shay Kutten, and Moti Yung.
\newblock The local detection paradigm and its applications to
  self-stabilization.

\bibitem{DFSConstruction1}
Baruch Awerbuch.
\newblock A new distributed depth-first-search algorithm.
\newblock {\em Information Processing Letters}, 20(3):147--150, 1985.

\bibitem{AwerbuchTransformer}
Baruch Awerbuch and George Varghese.
\newblock Distributed program checking: A paradigm for building
  self-stabilizing distributed protocols (extended abstract).
\newblock In {\em Proceedings of the 32Nd Annual Symposium on Foundations of
  Computer Science}, SFCS '91, pages 258--267, Washington, DC, USA, 1991. IEEE
  Computer Society.

\bibitem{chlamtac1987tree}
Imrich Chlamtac and Shay Kutten.
\newblock Tree-based broadcasting in multihop radio networks.
\newblock {\em Computers, IEEE Transactions on}, 100(10):1209--1223, 1987.

\bibitem{DFSConstruction2}
Isreal Cidon.
\newblock Yet another distributed depth-first-search algorithm.
\newblock {\em Inf. Process. Lett.}, 26(6):301--305, January 1988.

\bibitem{SelfStabDFS}
Zeev Collin and Shlomi Dolev.
\newblock Self-stabilizing depth-first search.
\newblock {\em Information Processing Letters}, 49(6):297 -- 301, 1994.

\bibitem{SnapStabCournier}
Alain Cournier, Stephane Devismes, Franck Petit, and Vincent Villain.
\newblock Snap-stabilizing depth-first search on arbitrary networks.
\newblock {\em Comput. J}, pages 268--280, 2006.

\bibitem{Cournier05:snap-stabilizing}
Alain Cournier, Stephane Devismes, and Vincent Villain.
\newblock A snap-stabilizing dfs with a lower space requirement.
\newblock In {\em In Seventh International Symposium on Self-Stabilizing
  Systems (SSS05}, pages 33--47, 2005.

\bibitem{Datta:DFTC}
Ajoy~K. Datta, Colette Johnen, Franck Petit, and Vincent Villain.
\newblock Self-stabilizing depth-first token circulation in arbitrary rooted
  networks.
\newblock {\em Distrib. Comput.}, 13(4):207--218, November 2000.

\bibitem{Dijkstra:1974}
Edsger~W. Dijkstra.
\newblock Self-stabilizing systems in spite of distributed control.
\newblock {\em Commun. ACM}, 17(11):643--644, November 1974.

\bibitem{Dolev}
Shlomi Dolev.
\newblock {\em Self-stabilization}.
\newblock MIT Press, 2000.

\bibitem{Dolev:lowerbound}
Shlomi Dolev, Mohamed~G. Gouda, and Marco Schneider.
\newblock Memory requirements for silent stabilization.
\newblock {\em Acta Informatica}, 36(6):447--462, 1999.

\bibitem{DolevIsraeliMoran}
Shlomi Dolev, Amos Israeli, and Shlomo Moran.
\newblock Self-stabilization of dynamic systems assuming only read/write
  atomicity.
\newblock {\em Distrib. Comput.}, 7(1):3--16, November 1993.

\bibitem{Dolev:1993}
Shlomi Dolev, Amos Israeli, and Shlomo Moran.
\newblock Self-stabilization of dynamic systems assuming only read/write
  atomicity.
\newblock {\em Distrib. Comput.}, 7(1):3--16, November 1993.

\bibitem{Even}
Shimon Even.
\newblock {\em Graph Algorithms}.
\newblock W. H. Freeman \& Co., New York, NY, USA, 1979.

\bibitem{SuomelaGoos}
Mika G{\"{o}}{\"{o}}s and Jukka Suomela.
\newblock Locally checkable proofs.
\newblock In {\em Proceedings of the 30th Annual {ACM} Symposium on Principles
  of Distributed Computing, {PODC} 2011, San Jose, CA, USA, June 6-8, 2011},
  pages 159--168, 2011.

\bibitem{HuangChen}
Shing-Tsaan Huang and Nian-Shing Chen.
\newblock Self-stabilizing depth-first token circulation on networks.
\newblock {\em Distributed Computing}, 7(1):61--66, 1993.

\bibitem{johnen97:DFTC}
Colette Johnen, Gianluigi Alari, Joffroy Beauquier, and Ajoy~K Datta.
\newblock Self-stabilizing depth-first token passing on rooted networks.
\newblock In {\em Distributed Algorithms}, pages 260--274. Springer, 1997.

\bibitem{Johnen95:DFTC}
Colette Johnen and Joffroy Beauquier.
\newblock Space-efficient, distributed and self-stabilizing depth-first token
  circulation.
\newblock In {\em In Proceedings of the Second Workshop on Self-Stabilizing
  Systems}, pages 4--1, 1995.

\bibitem{KatzPerry}
Shmuel Katz and Kenneth~J. Perry.
\newblock Self-stabilizing extensions for meassage-passing systems.
\newblock {\em Distributed Computing}, 7(1):17--26, 1993.

\bibitem{MSTVerification}
Amos Korman and Shay Kutten.
\newblock Distributed verification of minimum spanning trees.
\newblock In {\em Proceedings of the twenty-fifth annual ACM symposium on
  Principles of distributed computing}, PODC '06, pages 26--34, New York, NY,
  USA, 2006. ACM.

\bibitem{KormanKuttenSSMST}
Amos Korman, Shay Kutten, and Toshimitsu Masuzawa.
\newblock Fast and compact self stabilizing verification, computation, and
  fault detection of an mst.
\newblock In {\em Proceedings of the 30th Annual ACM SIGACT-SIGOPS Symposium on
  Principles of Distributed Computing}, PODC '11, pages 311--320, New York, NY,
  USA, 2011. ACM.

\bibitem{ProofLabeling}
Amos Korman, Shay Kutten, and David Peleg.
\newblock Proof labeling schemes.
\newblock {\em Distributed Computing}, 22(4):215--233, 2010.

\bibitem{Petit97highlyspace-efficient}
Franck Petit.
\newblock Highly space-efficient self-stabilizing depth-first token circulation
  for trees.
\newblock In {\em Euro-par'97 Parallel Processing, Proceedings LNCS}, pages
  47647--9. Springer-Verlag, 1997.

\bibitem{Petit:DFTC}
Franck Petit.
\newblock Fast self-stabilizing depth-first token circulation.
\newblock In {\em Proceedings of the 5th International Workshop on
  Self-Stabilizing Systems}, WSS '01, pages 200--215, London, UK, UK, 2001.
  Springer-Verlag.

\bibitem{PetitV97:DFTC}
Franck Petit and Vincent Villain.
\newblock Color optimal self-stabilizing depth-first token circulation.
\newblock In {\em ISPAN}, pages 317--323. IEEE Computer Society, 1997.

\bibitem{Petit97SelfStabTree}
Franck Petit and Vincent Villain.
\newblock Optimality and self-stabilization in rooted tree networks.
\newblock {\em Parallel Processing Letters}, 10(01):3--14, 2000.

\bibitem{PetitV00:DFTCM}
Franck Petit and Vincent Villain.
\newblock Self-stabilizing depth-first token circulation in asynchronous
  message-passing systems.
\newblock {\em Computers and Artificial Intelligence}, 19(5), 2000.

\bibitem{Petit2007:SnapStabOnTree}
Franck Petit and Vincent Villain.
\newblock Optimal snap-stabilizing depth-first token circulation in tree
  networks.
\newblock {\em Journal of Parallel and Distributed Computing}, 67(1):1 -- 12,
  2007.

\bibitem{Stomp}
F.A. Stomp.
\newblock Structured design of self-stabilizing programs.
\newblock In {\em Theory and Computing Systems, 1993., Proceedings of the 2nd
  Israel Symposium on the}, pages 167--176, Jun 1993.

\bibitem{Varghese00thefault}
George Varghese and Mahesh Jayaram.
\newblock The fault span of crash failures.
\newblock {\em Journal of the ACM}, 47:47--2, 2000.

\end{thebibliography}


\end{document}